\documentclass[letterpaper, 10 pt, conference]{ieeeconf}  

\IEEEoverridecommandlockouts                              

\usepackage{graphics} 
\usepackage{epsfig} 
\usepackage{times} 

\usepackage{amsmath,amsthm,amssymb}
\usepackage{mathtools}
\usepackage{mathrsfs}
\usepackage{subfiles}
\usepackage{color}
\usepackage{enumerate}
\usepackage[caption=false]{subfig}

\newtheorem{theorem}{Theorem}

\newtheorem{proposition}[theorem]{Proposition}

\newtheorem{definition}[theorem]{Definition}

\newtheorem{remark}[theorem]{Remark}

\newtheorem{goal}[theorem]{Goal}

\newcommand{\de}{\partial}
\newcommand{\R}{\mathbb{R}}
\newcommand{\N}{\mathbb{N}}

\newcommand{\dt}{\mathrm{d}t}
\newcommand{\dtau}{\mathrm{d}\tau}
\newcommand{\tr}{^T}

\newcommand{\mc}{\mathcal}

\title{\LARGE \bf
Passivity-Based Decentralized Control of Multi-Robot Systems\\
With Delays Using Control Barrier Functions
}

\author{Gennaro Notomista, Xiaoyi Cai, Junya Yamauchi, and Magnus Egerstedt
\thanks{This work was supported by ARL DCIST CRA W911NF-17-2-0181, Grant No. 1723943 from the U.S. National Science Foundation and JSPS KAKENHI Grant No. 18H01459.}%
\thanks{Gennaro Notomista, Xiaoyi Cai and Magnus Egerstedt are with the Institute for Robotics and Intelligent Machines, Georgia Institute of Technology, Atlanta, GA, USA {\tt\small \{g.notomista,xcai41,magnus\}@gatech.edu}}%
\thanks{Junya Yamauchi is with the Department of Systems and Control Engineering, Tokyo Institute of Technology, Tokyo, Japan {\tt\small yamauchi@sc.e.titech.ac.jp}}
}

\begin{document}

\maketitle
\thispagestyle{empty}
\pagestyle{empty}

\begin{abstract}

In this paper, we present a solution to the problem of coordinating multiple robots across a communication channel that experiences delays. The proposed approach leverages control barrier functions in order to ensure that the multi-robot system remains dissipative. This is achieved by encoding the dissipativity-preserving condition as a set invariance constraint. This constraint is then included in an optimization problem, whose objective is that of modifying, in a minimally invasive fashion, the nominal input to the robots. The formulated optimization problem is decentralized in the sense that, in order to be solved, it does not require the individual robots to have access to global information. Moreover, owing to its convexity, each robot can solve it using fast and efficient algorithms. The effectiveness of the proposed control framework is demonstrated through the implementation of a formation control algorithm in presence of delays on a team of mobile robots.

\end{abstract}

\section{INTRODUCTION}
\label{sec:intro}

Delayed communication networks pose a big challenge to coordinated control algorithms for multi-robot systems \cite{zampieri2008trends}. In order to implement consensus-like algorithms---where the robots are asked to agree on a common objective, such as the location where to meet \cite{olfati2004consensus} or the shape to assemble \cite{egerstedt2001formation}---or coverage control algorithms \cite{cortes2004coverage}, to name a few, the robots need to exchange information with their peers. Delays in transferring information can cause the performance of the algorithms to degrade, in terms of both convergence rate and stability \cite{olfati2004consensus}. 

In this paper, we propose a control strategy that is able to mitigate the effects of a delayed communication channel in multi-robot systems. This is achieved by modifying, in a \textit{minimally invasive} fashion, the nominal control input that is to be executed by the robots: this means that the modification of the robots' inputs only takes place if the effect of delays is about to compromise the desired performance of the system.

The problem of controlling multi-robot system with delays has been widely studied and different solutions have been proposed, ranging from predictor-based \cite{richard2003time} to passivity-based controllers \cite{chopra2006passivity}. In \cite{olfati2004consensus}, the authors analyze the performance of the consensus protocol with constant time delay and a fixed network topology. Their graph-theoretic approach shows the tight connection between convergence rate, eigenvalues of the graph Laplacian, and admissible time delays.

Different approaches, based on passivity, as well as other energy-based methods, are considered in  \cite{anderson1989bilateral,niemeyer1991stable,wohlers2017lumped}. Here, the so-called \textit{scattering transformation} is introduced in order to stabilize a system in the presence of any, although constant, communication delay. On a similar line of inquiry, in \cite{yamauchi2017passivity}, the authors employ the scattering transformation to make the interconnection of a network of robots with a human operator passive under constant time delays. It is worthwhile noticing that scattering-based approaches are suitable to make robots achieve synchronization tasks (such as consensus). However, these approaches cannot be employed to compensate for time delays when non-passive control laws---such as the formation controller in \cite{mesbahi2010graph}---are used.

Delays in communication channels can be viewed in terms of energy injected into a system, which can become non-passive, as discussed in~\cite{anderson1989bilateral}. Passivity theory allows us to analyze dynamical systems from an energetic point of view, so it is a very suitable design tool for dealing with systems with delays. In this paper, starting from dissipativity theory, we develop an optimization-based controller that is able to cope with delayed communication channels in multi-robot systems.

Finally, another passivity-based approach has been proposed in \cite{duindam2004port}, where the authors lay the foundations of the concept of \textit{energy tanks}.
Extensions of this technique are proposed in \cite{secchi2006position,secchi2007control,secchi2012bilateral,giordano2013passivity}. In all these works, the authors introduce an additional dissipative force on the system in order to prevent the energy tank from depleting---a condition that would introduce a singularity in their proposed approaches---so as to keep a positive \textit{passivity margin}\footnote{As in \cite{giordano2013passivity}, we refer to the passivity margin as the energy dissipated by the system over time.}. Using tools from nonlinear systems theory, we present a method that is also able to exploit the full passivity margin and, nonetheless, does not introduce any singularity.

\subsection{Main Contributions}
In view of what has been discussed in the previous section, we summarize here the main contributions presented in this paper. We propose a framework suitable to control multi-robot systems in presence of communication delays. The approach consists of an optimal way of exploiting the passivity margin of the system, where optimality stems from the minimally invasive nature of the method. In particular, the approach consists of an optimization-based method that is able to ensure the dissipativity of the multi-robot system. This is achieved by expressing the dissipativity condition as a set invariance constraint which is enforced by leveraging control barrier functions \cite{ames2014control}. More specifically, the use of control barrier functions allows us to formulate an optimization problem whose objective is to minimally modify the robots' nominal inputs, in order to ensure the dissipativity of the system.

The remainder of the paper is organized as follows. The next section establishes the multi-robot formation control problem under communication delays that motivates the proposed control method in the paper. Section~\ref{sec:background} will introduce some concepts from the analysis, modeling and control of nonlinear systems which will be used throughout the paper. Using the techniques recalled in Section~\ref{sec:background}, in Section~\ref{sec:cbfbaseddissipativity}, we formulate the optimization-based framework required to ensure the dissipativity of a dynamical system. Finally, the results of applying the developed approach to the coordinated control of a multi-robot system with communication delays are shown in Section~\ref{sec:experiments}.

\section{Problem Setup}
\label{sec:prob-setup}

\begin{figure}
\centering
\includegraphics[width=0.45\textwidth]{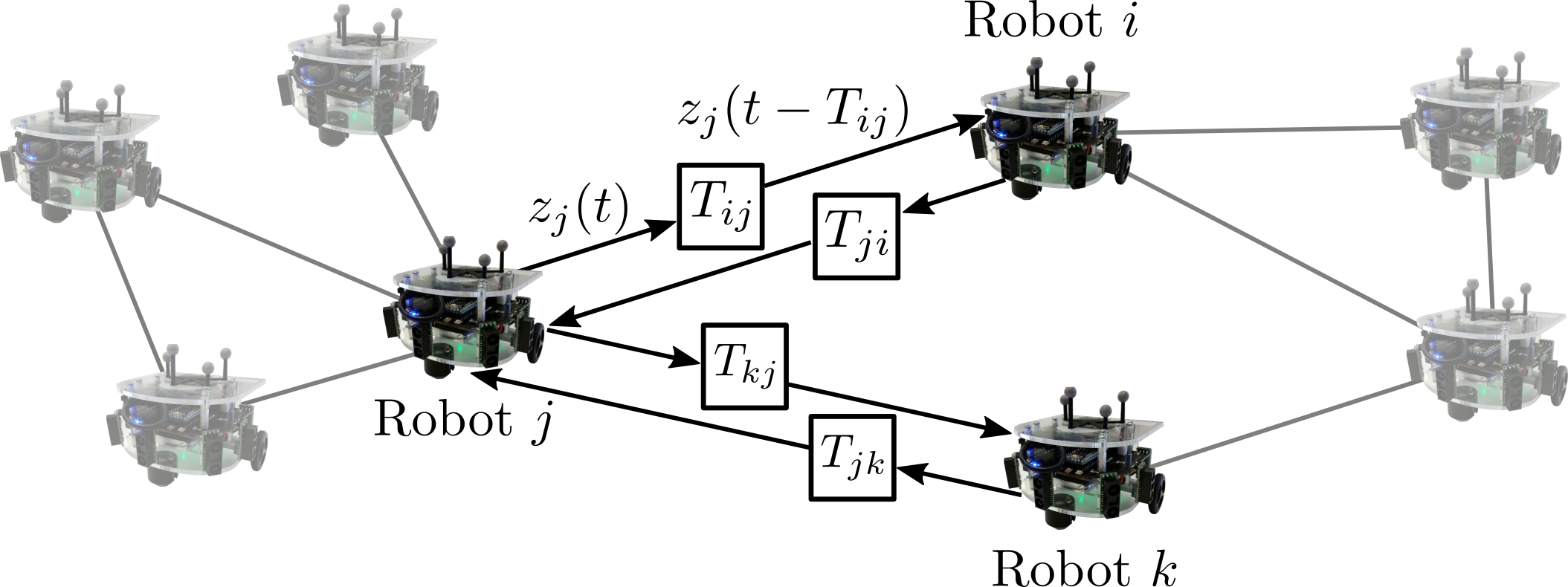}
\caption{Problem setup: a team of robots with delayed communication channel where each communication link may have a different value of time delay. In the depicted example, robot $i$ receives the delayed value of robot $j$'s position, $z_j(t-T_{ij})$.}
\label{comm_delay}
\end{figure}

Consider a group of $N$ planar mobile robots whose positions are denoted by $z_i\in\R^2$, $i\in\{1,...,N\}$. Each robot can be controlled using an acceleration input, and its dynamics are modeled as the following damped double integrator:
\begin{equation}
\label{eq:doubleint}
\ddot z_i = -a\dot z_i + u_i,
\end{equation}
where $a>0$ and $u_i$ is the acceleration input. Moreover, we assume all the robots are equipped with omnidirectional range sensors that allow them to measure their relative position with their neighbors, namely robot $i$ can measure the quantity $z_j-z_i$, when robot $j$ is within its sensing radius. The interactions among the robots can be described by a graph $\mc G=(\mc V, \mc E)$, where $\mc V\subset \N$ is the set of vertices of the graph, representing the robots, $\mc E \subseteq \mc V \times \mc V$ is the set of edges between the robots, encoding the adjacency relationships. We further assume that the graph is undirected, namely $(i,j)\in\mc E \Leftrightarrow (j,i)\in\mc E$.

In order to show the generality of the proposed approach, we consider a communication network between the robots where edges are characterized by different values of delay (see Fig.~\ref{comm_delay}). We denote by $T_{ij}\ge0$ the time delay on edge $(i,j)$ between robots $i$ and $j$. In case the communication link connecting robot $i$ to robot $j$ is considered to be physically the same as the one connecting robot $j$ to robot $i$, one can assume $T_{ij}=T_{ji}$. Nevertheless, the method proposed in this paper does not need this assumption to hold.

Let us assume that the robots are controlled to build a formation by maintaining specified distances on some edges in the set $\mc E$, with the distance $d_{ij}=d_{ji}$ corresponding to the edge $(i,j)$. We further assume the desired formation is both rigid and feasible, i.\,e., the desired shape can be maintained by only maintaining the specified distances, and there exist points $\xi_1,...,\xi_N\in\R^2$ such that $\left\|\xi_i-\xi_j\right\|=d_{ij}~\forall(i,j)\in\mc E$.
As a formation controller, we utilize the following weighted consensus protocol \cite{mesbahi2010graph}:
\begin{equation}
\label{eq:unom}
\hat u_i(t) = \sum\limits_{j\in \mc N_i}w_{ij}(z_i(t),z_j(t-T_{ij}))(z_j(t-T_{ij})-z_i(t))
\end{equation}
where $w_{ij}(z_i(t),z_j(t-T_{ij})) = \left\|z_i(t)-z_j(t-T_{ij})\right\|^2-d_{ij}^2$.
\begin{figure}
    \centering
    \includegraphics[trim={2cm 1cm 2cm 1cm}, clip, width=0.33\textwidth]{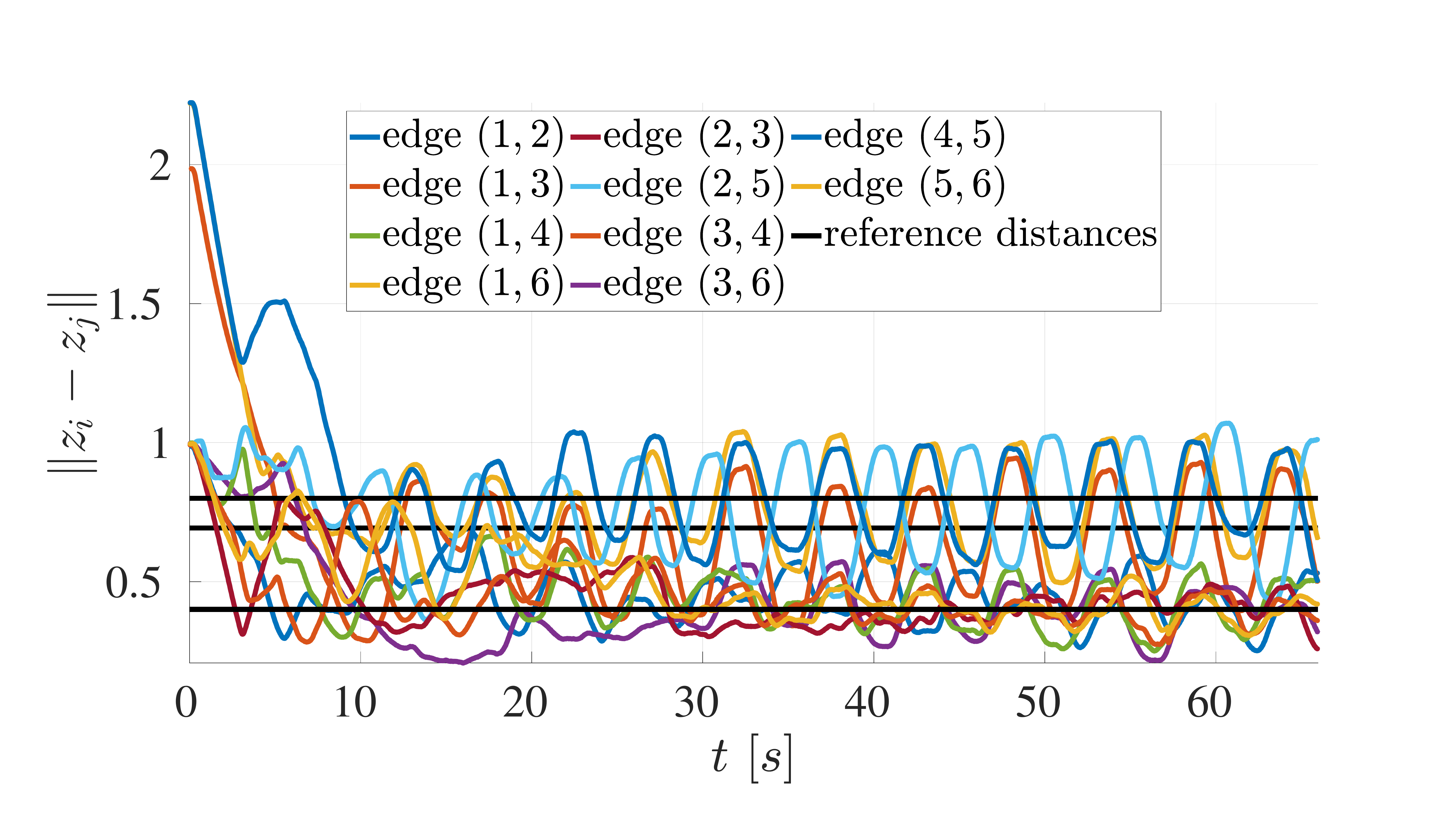}
    \caption{Values of the edge distances between the robots performing formation control with communication delays on the Robotarium \cite{pickem2017robotarium}. The values of the delays are different from edge to edge and are not symmetric, i.\,e., for edge $(i,j)$ between robots $i$ and $j$, $T_{ij}$ and $T_{ji}$ are not necessarily equal. As a consequence of the delayed communication channel, the robots are not able to maintain the desired inter-robot distances, depicted as black lines.}
    \label{fig:edgesnocbfs}
\end{figure}
Without any compensation for delays, the control input $\hat u_i$ can cause the multi-robot system to become unstable. Figure~\ref{fig:edgesnocbfs} shows the edge distances between the robots, while implementing the control protocol \eqref{eq:unom} on the Robotarium \cite{pickem2017robotarium}. As can be seen, because of the effect of an artificially delayed communication channel, the edge distances do not converge to the reference values for the inter-robot distances, depicted in black in the figure.

In the next section, we introduce the background material required to formulate a controller that takes the nominal input $\hat u_i$ and generates a new input which is as close as possible to the nominal one, but it is, at the same time, able to mitigate the oscillations shown in Fig.~\ref{fig:edgesnocbfs}, by preserving the dissipativity of each robot.

\section{BACKGROUND MATERIAL}
\label{sec:background}

The approach to the coordinated control of multi-robot systems with delayed communication networks we propose in this paper hinges on two concepts from the analysis and modeling of nonlinear dynamical systems: port-Hamiltonian systems and dissipativity theory, which are briefly recalled in Sections~\ref{subsec:phs}~and~\ref{subsec:passivity}, respectively. These methods will be combined with control barrier functions, which will be introduced in Section~\ref{subsec:cbfs}.

\subsection{Port-Hamiltonian Modeling}
\label{subsec:phs}

Initially derived in the context of analytical mechanics, the Hamiltonian approach to system modeling lends itself to the analysis of networked systems. Port-Hamiltonian systems are open dynamical systems which interact with other systems through input-output ports. The fact that the port-Hamiltonian modeling of a system is defined explicitly using stored and dissipated energies makes it a suitable analysis tool for the energy-based method we use in this paper.

A special case of port-Hamiltonian systems is the \textit{input-state-output port-Hamiltonian system} \cite{van2006port}:
\begin{equation}
\label{eq:phs}
\begin{cases}
\dot x = (J(x)-R(x))\dfrac{\de H}{\de x}\tr+g(x)u\\
y = g(x)\tr\dfrac{\de H}{\de x}\tr.
\end{cases}
\end{equation}
Here, $x\in\R^n$, $u,y\in\R^m$ are the state, input and output of the system, respectively, $H:\R^n\to\R$ is a continuously differentiable function, the \textit{Hamiltonian}, representing the total energy stored in the system, $J(x)=-J(x)\tr$ is the \textit{structure matrix} that expresses the interconnection structure of the system, and $R(x)=R(x)\tr\ge0$ is the \textit{resistive structure} modeling the dissipation in the system. $J(x)$, $R(x)$ and $g(x)$ are assumed to be continuously differentiable functions of the state $x$.

\subsection{Passivity and Dissipativity}
\label{subsec:passivity}

In this section, we review the definitions of passivity and dissipativity of dynamical systems, which allow us to analyze the behavior of an open system from an energetic point of view. This is our starting point to design coordinated control algorithms for multi-robot systems that are robust against communication delays.

\begin{definition}[Passivity, from \cite{khalil2015nonlinear}]
\label{def:passivity}
The system
\begin{equation}
\label{eq:nls}
\begin{cases}
\dot x = f(x,u)\\
y = c(x,u),
\end{cases}
\end{equation}
where $x\in\R^n$ is the state, $u,y\in\R^m$ are the input and output, respectively, is passive if there exists a positive definite, continuously differentiable storage function ${V:\R^n\to\R}$ such that, for all $x$ and $u$,
\begin{equation*}
\dot V = \dfrac{\de V}{
\de x} f(x,u) \le u\tr y.
\end{equation*}
The system is called lossless if $\dot V = u\tr y$.
\end{definition}

Due to the point-wise nature of Definition~\ref{def:passivity}, enforcing a passivity condition on multi-robot systems to ensure robustness to delays may be too restrictive. That is why we make use of the more general concept of dissipativity in its integral form, defined below.

\begin{definition} [Dissipativity, from \cite{willems1972dissipative}] A dynamical system \eqref{eq:nls} is dissipative if there exists a nonnegative storage function ${S:\R^n\to\R}$ such that, for all $t_0,t_1\in\R_+$ and input signal $u(t)$,
\begin{equation}
\label{eq:dissipativity}
S(x(t_1)) \le S(x(t_0)) + \int_{t_0}^{t_1} w(u(t),y(t))\dt,
\end{equation}
where the locally integrable function $w:\R^m\times\R^m\to\R$ is called the supply rate. Analogously, one can define lossless dissipativity if \eqref{eq:dissipativity} holds with equality.
\label{def:dissipativity}
\end{definition}

The inequality \eqref{eq:dissipativity} says that, if a system is dissipative, the value of the storage function $S$ at each point in time $t_1$ cannot exceed its initial value at time $t_0$ augmented by the supplied energy $\int_{t_0}^{t_1} w(u(t),y(t))\dt$.

\begin{remark}
\label{rmk:passivitydissipativity}
Note that, if the system \eqref{eq:nls} is passive with storage function $V$, then it is dissipative with the storage function $S=V$ and supply rate function given by
\begin{equation}
\label{eq:supplyrateinnerproduct}
w:(u(t),y(t))\mapsto u(t)\tr y(t).
\end{equation}
\end{remark}

As mentioned above, the connection between passivity and dissipativity---pointed out in Remark~\ref{rmk:passivitydissipativity}, see also \cite{bao2007process}---suggests a relaxation from the point-wise definition of passivity condition to the integral definition of the dissipativity condition. This is the approach we will follow in this paper in order to formulate an optimization-based strategy that is able to ensure the dissipativity of an open dynamical system. A similar strategy led to the idea of energy tanks (see, e.\,g., \cite{giordano2013passivity}), whose connection with our method will be discussed in Section~\ref{sec:cbfbaseddissipativity}.

\begin{remark}
Notice that a port-Hamiltonian system \eqref{eq:phs} is passive with storage function $V=H$. Therefore, by Remark~\ref{rmk:passivitydissipativity}, it is also dissipative with storage function $S=H$ and supply rate given by \eqref{eq:supplyrateinnerproduct} (see, e.\,g., \cite{ortega2002interconnection}). Indeed, taking the time derivative of $H$, one obtains
\begin{equation}
\label{eq:passivityphs}
\dot H = -\dfrac{\de H}{\de x} R(x) \dfrac{\de H}{\de x}\tr + u\tr y = -D(x)+u\tr y \le u\tr y,
\end{equation}
since $J(x)$ is skew-symmetric, and
\begin{equation*}
D(x) \triangleq \dfrac{\de H}{\de x} R(x) \dfrac{\de H}{\de x}\tr \ge 0,
\end{equation*}
as the matrix $R(x)$ is positive definite. Then, integrating \eqref{eq:passivityphs} over the interval $[t_0,t_1]$ yields
\begin{equation}
\label{eq:dissipativityphs}
\begin{aligned}
H(x(t_1)) &= H(x(t_0)) + \int_{t_0}^{t_1} u(t)\tr y(t)\dt - \int_{t_0}^{t_1} D(x)\dt\\
&\le H(x(t_0)) + \int_{t_0}^{t_1} u(t)\tr y(t)\dt.
\end{aligned}
\end{equation}
\end{remark}

The dissipativity condition introduced in this section is encoded through integral inequalities, such as \eqref{eq:dissipativityphs}. In the next section, we present control barrier functions, which are helpful to formulate algebraic constraints on the control input variable $u$ to be enforced on the system in order to preserve its dissipativity.

\subsection{Control Barrier Functions}
\label{subsec:cbfs}

As mentioned in Section~\ref{sec:intro}, the use of energy tanks---introduced in \cite{duindam2004port}---allows non-passive actions to be executed when enough passivity margin has been accumulated (see \cite{giordano2013passivity} for a more detailed analysis). The singularity introduced in the system by the presence of the energy tank corresponds to the situation in which the system becomes non-dissipative. In this case, no action that would lead to a further decrease of the total energy of the system can be implemented without losing the dissipativity property.

In this section, we introduce a method to overcome this issue by ensuring that the the system always remains dissipative. This will be achieved leveraging a tool from nonlinear systems theory called control barrier functions \cite{ames2014control}.

Ensuring that the dissipativity condition in \eqref{eq:dissipativityphs} holds lends itself to be formulated as a set forward invariance condition. The forward invariance of a set can be encoded by the definition of a control barrier function. For our work, we refer to the variation on control barrier functions introduced in \cite{xu2015robustness} and their time-varying version in \cite{2019arXiv190305810N}, whose definitions are recalled in the following.

\begin{definition}[Zeroing control barrier function, from \cite{xu2015robustness}]
Consider a dynamical system in control affine form 
	\begin{equation}
	\label{eq:canls}
	\dot x = f(x) + g(x) u,
	\end{equation}
with $f$ and $g$ locally Lipschitz continuous vector fields, and a set $\Omega$ defined as the zero superlevel set of a continuously differentiable function ${h:\R^n\to\R}$, i.\,e.,
\begin{equation}\label{eq:safeset}
	\Omega = \{ x\in \R^n~|~h(x)\ge0 \}.
\end{equation}
The function $h$ is called a zeroing control barrier function, if there exists a locally Lipschitz extended class $\mc K$ function $\gamma$ such that
\begin{equation}\label{eq:uzcbf}
	\sup_{u\in \R^m} \left\{L_f h(x) + L_g h(x) u + \gamma(h(x))\right\} \ge 0\quad\forall x\in \R^n,
\end{equation}
where $L_f h(x)$ and $L_g h(x)$ denote the Lie derivatives of $h$ in the directions of the vector fields $f$ and $g$, respectively.
\end{definition}

The following theorem summarizes two important properties of control barrier functions.

\begin{theorem}[See \cite{xu2015robustness} and \cite{arxiv:extendedacc}]\label{thm:zcbfproperties}
	Given a dynamical system in control affine form \eqref{eq:canls}, where $x\in\R^n$ and $u\in\R^m$ denote the state and the input, respectively, $f$ and $g$ are locally Lipschitz, and a set $\Omega\subset\R^n$ defined by a continuously differentiable function $h$ as in \eqref{eq:safeset}, any Lipschitz continuous controller $u$ such that \eqref{eq:uzcbf} holds renders the set $\Omega$ forward invariant and asymptotically stable, i.\,e.,
	\begin{align*}
	&x(0)\in \Omega \Rightarrow x(t)\in \Omega \quad \forall t\ge0\\
	&x(0)\notin \Omega \Rightarrow x(t)\rightarrow\in \Omega \quad \text{as}~t\to\infty,
	\end{align*}
	where $x(0)$ denotes the state $x$ at time $t=0$.
\end{theorem}

In order to handle the time-varying constraints coming from the dissipativity condition, in the following, we recall the definition of time-varying control barrier function, introduced in \cite{2019arXiv190305810N}.

\begin{definition}[Time-varying zeroing control barrier fucntions \cite{2019arXiv190305810N}]
\label{def:tvzcbf}
Given a function $h : \R^n \times \R_+ \mapsto \R$, continuously differentiable in both its arguments, consider a dynamical system in control affine form \eqref{eq:canls}, where $x\in\R^n$ and $u\in\R^m$ denote system state and input, respectively, $f$ and $g$ are locally Lipschitz, and the set $\Omega=\left\{ x \in \R^n ~\vert~ h(x,t)\ge0 \right\}$. The function $h$ is a time-varying zeroing control barrier function defined on $\mathcal \R^n \times \R_+$, if there exists a locally Lipschitz extended class $\mathcal K$ function $\gamma$ such that, $\forall x \in \mathcal \R^n$, $\forall t \in \mathcal \R_+$,
\begin{equation}
\label{eq:utvzcbf}
\sup_{u \in \R^m} \left\{ \frac{\partial h}{\partial t} + L_f h(x,t) + L_g h(x,t)\,u + \gamma(h(x,t))\right\} \ge 0.
\end{equation}
\end{definition}

\begin{remark}
Notice that a port-Hamiltonian system \eqref{eq:phs} is in control affine form \eqref{eq:canls}, with
\begin{equation}
\label{eq:fofx}
f(x) = (J(x)-R(x))\dfrac{\de H}{\de x}\tr.
\end{equation}
Therefore, the results developed for control barrier functions apply directly to port-Hamiltonian systems as well.
\end{remark}

The results recalled in this section are used in the following section in order to achieve a \textit{passivation objective} \cite{ortega2002interconnection}, corresponding to ensuring that the dissipativity condition \eqref{eq:dissipativityphs} holds at each point in time.

\section{PASSIVATION USING\\CONTROL BARRIER FUNCTIONS}
\label{sec:cbfbaseddissipativity}

In this section, we propose an optimization-based approach, applicable to a port-Hamiltonian system \eqref{eq:phs}, that is able to ensure that the dissipativity condition \eqref{eq:dissipativityphs} is always satisfied. A similar objective has been pursued in \cite{duindam2004port}, where the authors use a controlled energy transfer between the system and an energy tank, which stores the energy dissipated by the system, in order to be able to implement non-passive actions without violating the dissipativity of the system. Examples of the application of this technique can be found, e.\,g., in \cite{secchi2012bilateral,giordano2013passivity}. However, in these approaches there is no notion of optimality in preventing the energy stored in the energy tank from becoming negative, but rather an artificial dissipative action is introduced in order to accomplish this goal. In the strategy we develop in this paper, we propose an optimization-based controller that modifies, in a minimally invasive fashion, the input to the system, in such a way that the dissipativity condition \eqref{eq:dissipativityphs} is guaranteed to be satisfied.

To this end, let us introduce the following \textit{passivation objective}, which is based on the one defined in \cite{ortega2002interconnection}.
\begin{goal}[Passivation objective]
\label{goal:passivation}
Given the system \eqref{eq:phs}, with storage function $H$, and a nominal control action $\hat u$, find the control input $u = \tilde u + v$, such that $\tilde u$ is as close to $\hat u$ as possible, provided that the system is dissipative with supply rate $v(t)\tr y(t)$. We can formalize this goal in the following optimization program:
\begin{align*}
\label{eq:goal}
\min_{\tilde u} &~\|\tilde u-\hat u\|^2\\
\mathrm{s.t.} &~H(x(t)) \le H(x(t_0)) + \int_{t_0}^{t} v(\tau)\tr y(\tau)\dtau.
\end{align*}
\end{goal}

With this objective, and in view of what has been presented in Section~\ref{sec:background}, let us define the following time-varying control barrier function:
\begin{equation}
\label{eq:ourcbf}
h(x,t) \triangleq \int_{t_0}^t \left(D(x(\tau)) - \tilde u(\tau)\tr y(\tau)\right)\dtau.
\end{equation}
With this definition of $h(x,t)$, and substituting \eqref{eq:fofx} into the inequality \eqref{eq:utvzcbf} that is required to ensure forward invariance of the set $\Omega = \{x\in\R^n~|~h(x,t)\ge 0\}$, we obtain
\begin{equation}
\label{eq:affineineq}
D(x)-y\tr\tilde u+\gamma(h(x,t))\ge0.
\end{equation}
The next proposition establishes the validity of the control barrier function defined in \eqref{eq:ourcbf} and, at the same time, presents an optimization-based controller that is able to achieve Goal~\ref{goal:passivation}.

\begin{proposition}
\label{prop:cbfforpassivation}
The control input $\tilde u^\ast$, solution of
\begin{equation}
\label{eq:qp1}
\begin{aligned}
\min_{\tilde u} &~\|\tilde u-\hat u\|^2\\
\mathrm{s.t.} &~D(x)-y\tr\tilde u+\gamma(h(x,t))\ge0,
\end{aligned}
\end{equation}
renders the system
\begin{equation}
\label{eq:phs+}
\begin{cases}
\dot x = (J(x)-R(x))\dfrac{\de H}{\de x}\tr+g(x)\tilde u+g(x)v\\
y = g(x)\tr\dfrac{\de H}{\de x}\tr
\end{cases}
\end{equation}
dissipative with supply rate $v(t)\tr y(t)$.
\end{proposition}
\begin{proof}
From Lemma 12 in \cite{2019arXiv190305810N} (the time-varying version of Theorem~\ref{thm:zcbfproperties}), we can conclude that, enforcing the constraint $D(x)-y\tr\tilde u+\gamma(h(x,t))\ge0$ ensures that the set $\Omega = \{x\in\R^n~|~h(x,t)\ge 0\}$ is forward invariant. This is equivalent to saying that, for the system \eqref{eq:phs+}, the following holds $\forall t\ge t_0$:
\begin{equation}
\label{eq:cbfmeaning}
h(x,t) = \int_{t_0}^t \left(D(x(\tau)) - \tilde u(\tau)\tr y(\tau)\right)\dtau \ge0.
\end{equation}
Now, consider the energy balance \eqref{eq:dissipativityphs} for the system \eqref{eq:phs+} over the time interval $[t_0,t]$, for an arbitrary time instant $t$:
\begin{align*}
H(x(t)) =& H(x(t_0)) + \int_{t_0}^{t} \tilde u(\tau)\tr y(\tau)\dtau\\
&+\int_{t_0}^{t} v(\tau)\tr y(\tau)\dtau - \int_{t_0}^{t} D(x(\tau))\dtau.
\end{align*}
Using this equality and the condition \eqref{eq:cbfmeaning}, yields
\begin{equation*}
H(x(t)) \le H(x(t_0)) + \int_{t_0}^t v(\tau)\tr y(\tau)\dtau \quad \forall t\ge t_0,
\end{equation*}
which, from Definition~\ref{def:dissipativity}, means that the system is dissipative with storage function $H$ and supply rate $v(t)\tr y(t)$.
\end{proof}

\begin{remark}
The optimization program \eqref{eq:qp1} is a quadratic program (QP): therefore, it can be solved very efficiently (see, for instance, methods in \cite{boyd2004convex}) in order to evaluate the control input which is closest to the desired one $\hat u$ and which satisfies, at the same time, the dissipativity constraint \eqref{eq:affineineq}.
\end{remark}

\begin{remark}[Feasibility of the QP \eqref{eq:qp1}]
\label{rmk:feasibility} Notice that the QP in \eqref{eq:qp1} is feasible if $\frac{\de H}{\de x}g(x) \neq 0$, which is a relative-degree-1 condition on the control barrier function \eqref{eq:ourcbf}. In case this condition is not satisfied, techniques to handle high relative degree control barrier functions can be employed \cite{nguyen2016exponential,2019arXiv190305810N}.
\end{remark}

In practice, it is desirable to have direct control over the energy dissipated by the system, given by $\int_{t_0}^t D(x(\tau))\dtau$, as well (see, e.\,g., \cite{giordano2013passivity,secchi2012bilateral}). Similarly to what has been done for energy tanks to exploit positive passivity margin \cite{giordano2013passivity}, in this section we introduce an alternative approach which allows us to directly modify both the system input, $\tilde u$, and its resistive structure $R(x)$.

\begin{figure*}
    \centering
    \subfloat[][]{\label{subfig:edges}\includegraphics[trim={1cm 1cm 4cm 1cm}, clip,width=0.33\textwidth]{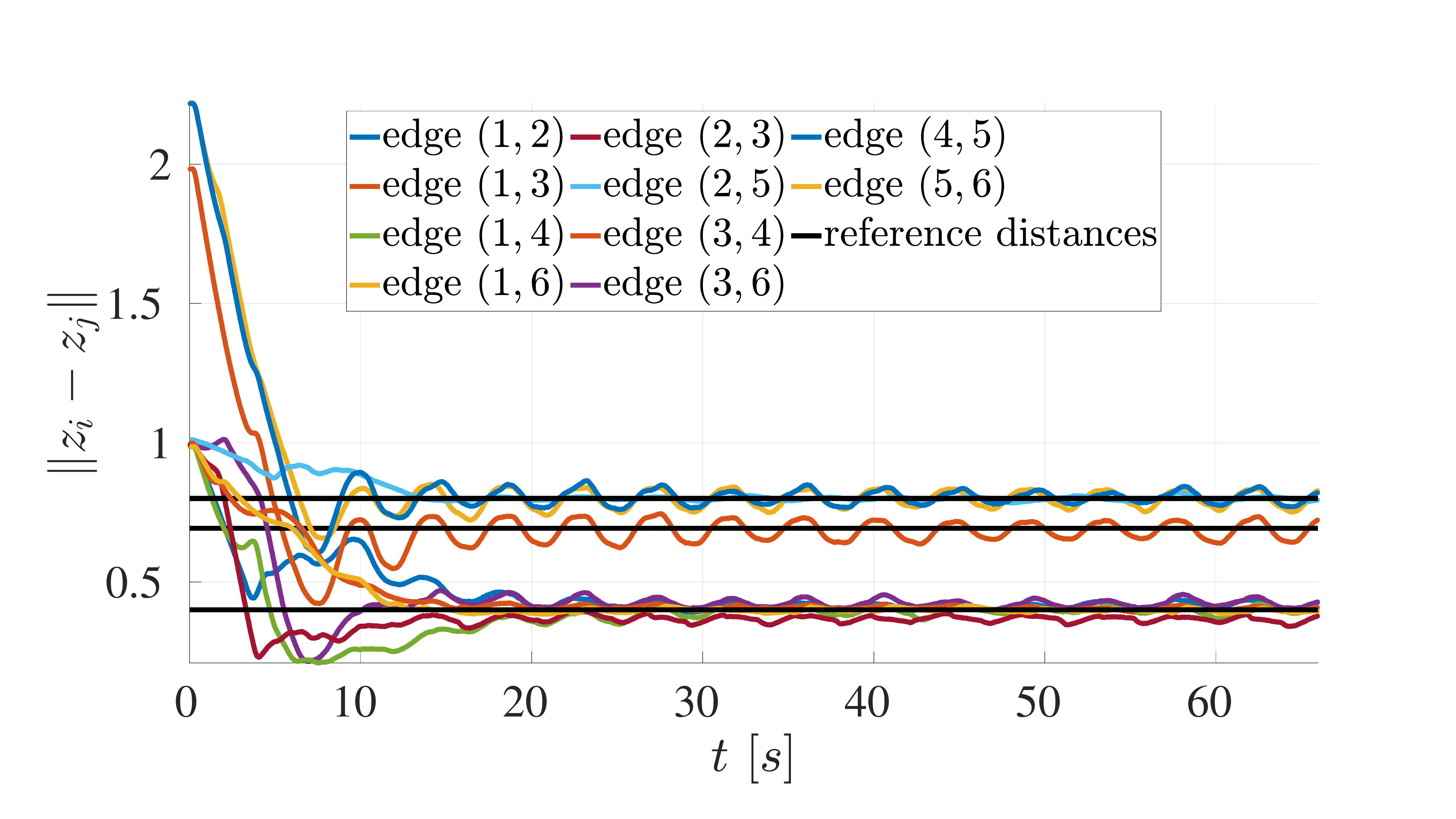}}\hfill
    \subfloat[][]{\label{subfig:cbfs}\includegraphics[trim={3cm 1cm 4cm 1cm}, clip,width=0.33\textwidth]{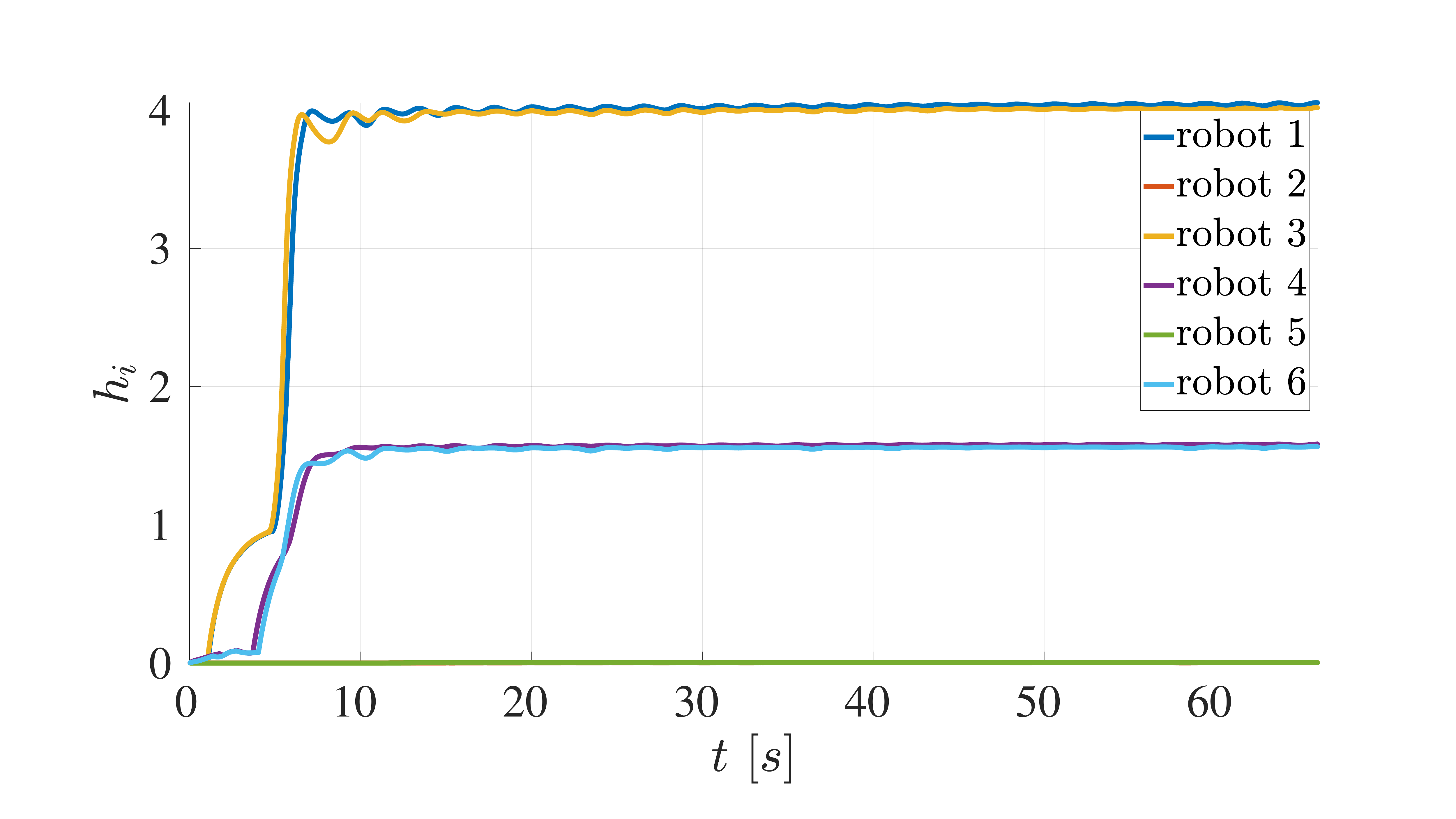}}\hfill
    \subfloat[][]{\label{subfig:inputs}\includegraphics[trim={1cm 1cm 4cm 1cm}, clip,width=0.33\textwidth]{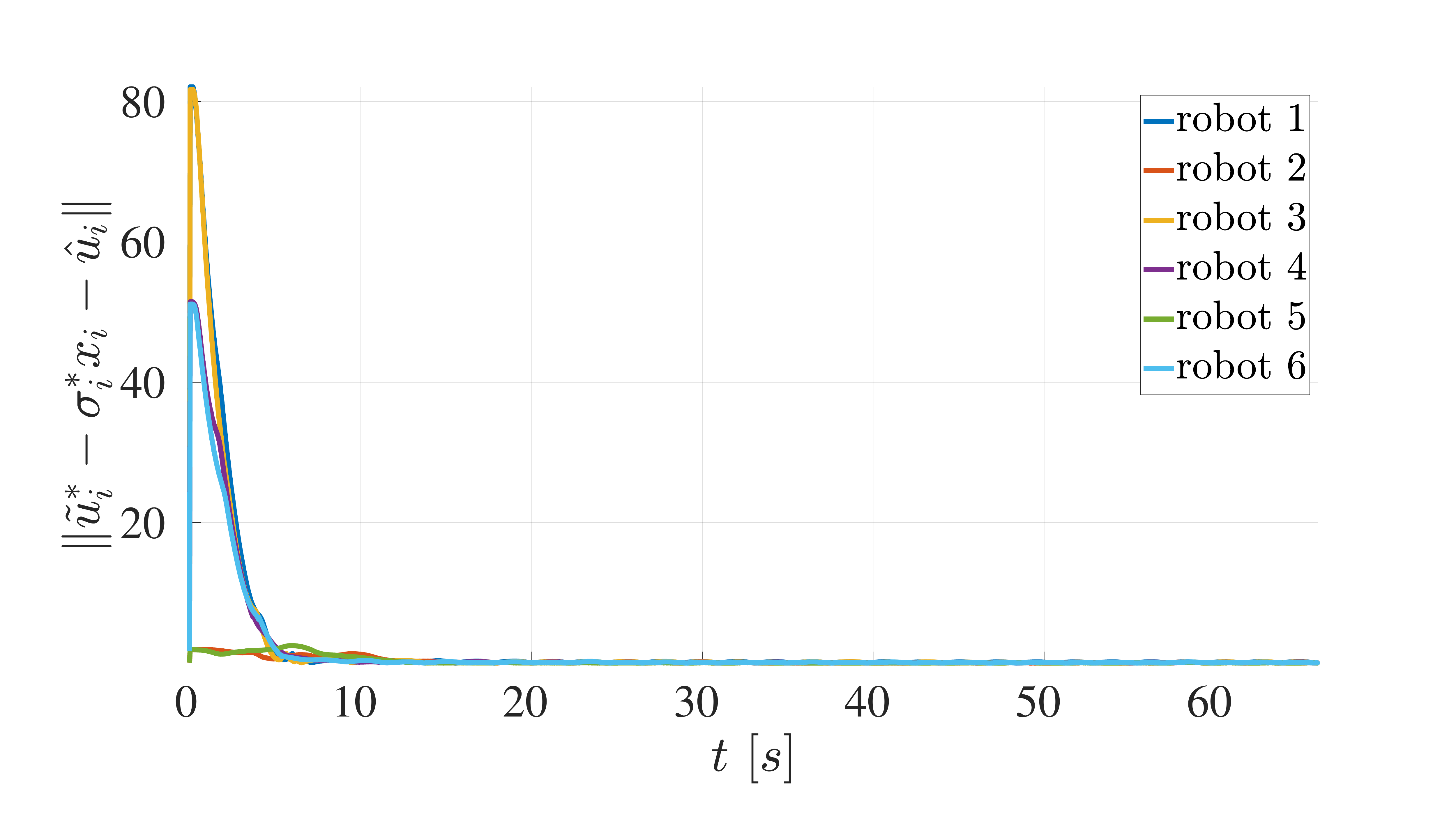}}
    \caption{Results of the application of the proposed approach to a multi-robot system consisting of six differential drive robots deployed on the Robotarium \cite{pickem2017robotarium}, implementing a non-passive formation control algorithm over a delayed communication channel. In Fig.~\protect\ref{subfig:edges}, the values of the edge distances between the robots converge to the desired values $d_{ij}$, depicted as black lines (compare this result with Fig.~\ref{fig:edgesnocbfs} where no action was taken to compensate for the delays). Figures~\protect\ref{subfig:cbfs}~and~\ref{subfig:inputs} show, for each robot, the values of the control barrier functions and the difference between nominal and executed inputs, respectively. As can be seen, at the beginning, the nominal inputs are changed significantly in order to keep a positive value of the barrier function.}
    \label{fig:withcbfs}
\end{figure*}

To this end, consider the new control input
\begin{equation}
\label{eq:uv}
u = \tilde u - \sigma y + v = \tilde u - \sigma g(x)\tr\dfrac{\de H}{\de x}\tr + v,
\end{equation}
where $\sigma$ is a positive real number, representing an additional damping coefficient. With this choice of input, the system \eqref{eq:phs+} can be written as:
\begin{equation}
\label{eq:phs++}
\begin{cases}
\dot x = (J(x)-\widetilde R(x))\dfrac{\de H}{\de x}\tr+g(x)\tilde u+g(x)v\\
y = g(x)\tr\dfrac{\de H}{\de x}\tr,
\end{cases}
\end{equation}
where $\widetilde R(x)=R(x)+\sigma g(x)g(x)\tr\ge0$. Defining $h$ analogously to \eqref{eq:ourcbf} as follows
\begin{equation*}
\label{eq:ourcbf2}
h(x,t) \triangleq \int_{t_0}^t \left(\widetilde D(x(\tau)) - \tilde u(\tau)\tr y(\tau)\right)\dtau,
\end{equation*}
the differential constraint in \eqref{eq:utvzcbf} becomes
\begin{equation}
\label{eq:affineineq2}
\widetilde D(x)-y\tr\tilde u+\gamma(h(x,t))\ge0,
\end{equation}
where
\begin{align*}
\widetilde D(x) &\triangleq \dfrac{\de H}{\de x}\widetilde R(x)~\dfrac{\de H}{\de x}\tr = D(x)+\sigma\left\|g(x)\tr\dfrac{\de H}{\de x}\tr\right\|^2\\
&= D(x)+\sigma\left\|y\right\|^2.
\end{align*}

The expression in \eqref{eq:affineineq2} is affine both in $\tilde u$ and in $\sigma$. Thus, proceeding as before, we can define the following QP:
\begin{equation}
\label{eq:qp2}
\begin{aligned}
\min_{\tilde u,\sigma} &~\|\tilde u-\hat u\|^2+\kappa|\sigma|^2\\
\mathrm{s.t.} &~\widetilde D(x)-y\tr\tilde u+\gamma(h(x,t))\ge0,
\end{aligned}
\end{equation}
where $\kappa$ allows us to combine the costs of \textit{transparency} (represented by $\|\tilde u-\hat u\|$) and dissipated energy (related to $|\sigma|$). Feasibility considerations similar to those in Remark~\ref{rmk:feasibility} hold for the optimization problem \eqref{eq:qp2} as well.

\begin{remark}
The proposed approach keeps transparency and passivity within a single optimization program which, as the one defined in \eqref{eq:qp1}, can be solved efficiently even under real-time constraints. Among other solutions that have been proposed in order to combine transparency and passivity, there is the two-layer paradigm presented in \cite{franken2011bilateral}, where, differently from what we developed in this paper, keeps transparency and passivity in two separate layers.
\end{remark}

We conclude by stating the following proposition---whose proof is similar to that of Proposition~\ref{prop:cbfforpassivation}---which summarizes the results obtained in this section.
\begin{proposition}
The control input $\tilde u^\ast$, solution of \eqref{eq:qp2}, guarantees that Goal~\ref{goal:passivation} is accomplished, by ensuring that the port-Hamiltonian system \eqref{eq:phs++} is dissipative with supply rate $v(t)\tr y(t)$.
\end{proposition}

The method developed in this section ensures that a port-Hamiltonian system remains dissipative even under non-passive control actions. In the following section, this control framework will be applied to a formation control algorithm and deployed on a team of mobile robots with delayed communication channel.

\section{EXPERIMENTS}
\label{sec:experiments}

We can now revisit the formation control problem formulated in Section~\ref{sec:prob-setup}. In order to apply the control framework derived in Section~\ref{sec:cbfbaseddissipativity}, let $x_i=\dot z_i$ be the velocity of robot $i$, and $H(x_i)=\frac{1}{2}x_i^Tx_i$. Then, the robot dynamics \eqref{eq:doubleint} can be rewritten in port-Hamiltonian form \eqref{eq:phs} as follows:
\begin{equation*}
\label{eq:phs_exp}
\begin{cases}
\dot x_i = -ax_i+u_i\\
y_i = x_i.
\end{cases}
\end{equation*}
Here, the structure matrix $J(x_i)$ vanishes, the resistive structure is given by $R(x_i)=aI$, $I$ being a $2\times 2$ identity matrix, $g(x_i)=I$, and $\frac{\de H}{\de x_i}^T=x_i$. 
Notice that $D(x_i) = a\left\|x_i\right\|^2\ge0$, that shows the robot model is dissipative.

Considering the control input $u_i=\tilde u_i-\sigma_i x_i+v_i$, as in \eqref{eq:uv}, we can wrap the optimization problem \eqref{eq:qp2} around the nominal control input $\hat u_i$ in \eqref{eq:unom} in order to get the control input $\tilde u_i^\ast$---that is as close as possible to $\hat u_i$---and the smallest damping coefficient $\sigma_i^\ast$ that guarantee that the robots remain dissipative. In Fig.~\ref{fig:withcbfs}, the results of the implementation of the proposed method to the same scenario introduced in Section~\ref{sec:prob-setup} are presented. Time delays $T_{ij}$ have been artificially introduced, and their values have been randomly chosen between 0s and 0.333s. In particular, the edge distances between the robots, the values of the control barrier functions related to each robot, and the difference between the nominal inputs and the ones executed by the robots are shown. In this experiment, the value of the external input $v_i$ has been set to zero for all the robots.

\begin{figure}
    \centering
    \includegraphics[trim={12cm 7cm 12cm 7cm}, clip, width=0.3\textwidth]{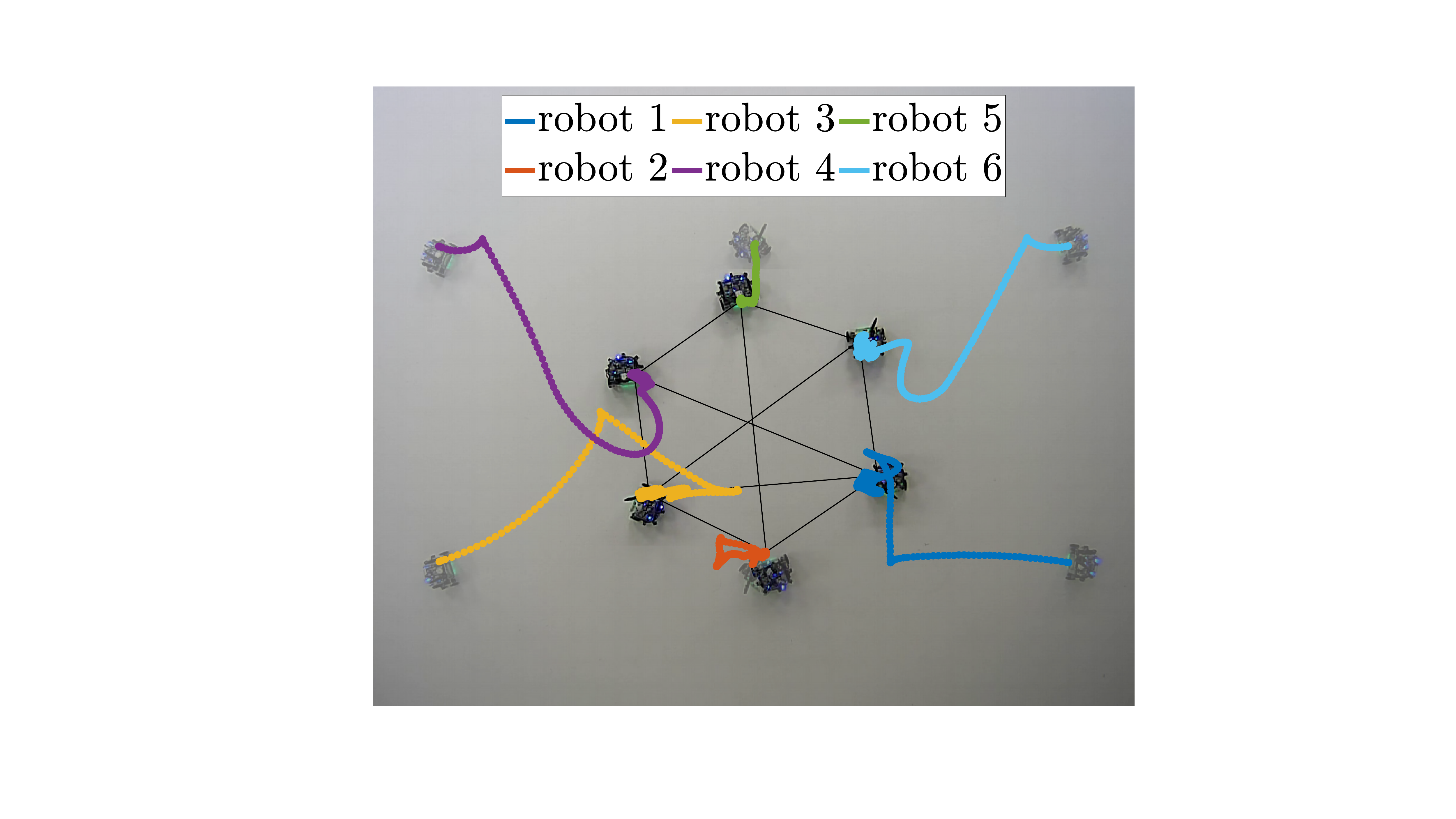}
    \caption{Trajectories of six differential drive robots deployed on the Robotarium \cite{pickem2017robotarium} to execute the controller obtained by solving \eqref{eq:qp2}. The robots achieve the desired formation over a communication network that experiences delays. Two snapshots taken at the beginning and at the end of the experiment, showing the initial and final positions of the robots, respectively, are overlaid on each other.}
    \label{fig:traj}
\end{figure}
Fig.~\ref{fig:traj} shows the trajectories of six differential drive robots on the Robotarium \cite{pickem2017robotarium} achieving the desired formation, despite the presence of non-uniform (i.\,e., varying from edge to edge) and non-symmetric (i.\,e., $T_{ij}$ might be different from $T_{ji}$) delays over the communication links between the robots.

\section{CONCLUSIONS}

In this paper, we have presented an optimization-based strategy suitable for the coordinated control of multi-robot systems with a delayed communication channel. The approach leverages control barrier functions, and dissipativity theory, in order to formulate a dissipativity-preserving controller. The proposed method is minimally invasive, in the sense that the modification of the nominal robots' input only happens if the effect of delays is about to compromise the desired performance of the system. Moreover, the resulting optimization program can be efficiently solved even on robotic platforms with limited computational power. The effectiveness of the proposed strategy to compensate for communication delays has been demonstrated on a team of mobile robots achieving formation control.

\bibliographystyle{IEEEtran}
\bibliography{bib/IEEEabrv,bib/IEEEexample}

\addtolength{\textheight}{-12cm}   





\end{document}